\documentclass[11pt]{article} 

\usepackage[utf8]{inputenc} 

\usepackage{geometry} 
\usepackage{enumerate}
\usepackage{amsthm}

\usepackage{graphicx,color}    
\usepackage{epsfig}

\newtheorem{theorem}{Theorem}

\newtheorem{lemma}[theorem]{Lemma}
\newtheorem{observation}[theorem]{Observation}
\newtheorem*{notation}{Notation}
\newtheorem{corollary}[theorem]{Corollary}

\usepackage{amsmath}
\usepackage{amssymb}

\newcommand\efface[1]{}

\title{Incidence coloring game and arboricity of graphs}
\author{Clément Charpentier and \'Eric Sopena\\
  \mbox{}\\
  {\small Univ. Bordeaux, LaBRI, UMR5800, F-33400 Talence}\\
  {\small CNRS, LaBRI, UMR5800, F-33400 Talence}}
\date{\today} 

\begin{document}

\maketitle


\begin{abstract}
An incidence of a graph $G$ is a pair $(v,e)$ where $v$ is a vertex of $G$ and $e$ an edge incident to $v$. 
Two incidences $(v,e)$ and $(w,f)$ are adjacent whenever $v = w$, or $e = f$, or $vw = e$ or $f$.
The incidence coloring game [S.D. Andres, The incidence game chromatic number, Discrete Appl. Math. 157 (2009), 1980–1987]
is a variation of the ordinary coloring game where the two players, Alice and Bob, alternately color the incidences of a graph, using
a given number of colors, in such a way that adjacent incidences get distinct colors. If the whole graph is colored then
Alice wins the game otherwise Bob wins the game.
The incidence game chromatic number $i_g(G)$ of a graph $G$ is the minimum number of colors for which Alice has a winning
strategy when playing the incidence coloring game on $G$. 

Andres proved that 
$i_g(G) \le 2\Delta(G) + 4k - 2$ for every $k$-degenerate graph $G$.
We show in this paper that $i_g(G) \le \lfloor\frac{3\Delta(G) - a(G)}{2}\rfloor + 8a(G) - 2$ for every graph $G$, 
where $a(G)$ stands for the arboricity of $G$, thus improving
the bound given by Andres since $a(G) \le k$ for every $k$-degenerate graph $G$. Since there exists graphs with $i_g(G) \ge \lceil\frac{3\Delta(G)}{2}\rceil$, the multiplicative constant of our bound is best possible. 
\end{abstract}

\noindent
{\bf Keywords:} Arboricity; Incidence coloring; Incidence coloring game; Incidence game chromatic number.

\section{Introduction}

All the graphs we consider are finite and undirected. For a graph $G$, we denote by $V(G)$, $E(G)$ and $\Delta(G)$ its vertex set, edge set and maximum degree, respectively. Recall that a graph is \emph{$k$-denegerate} if all of its subgraphs have minimum degree at most $k$.

The \textit{graph coloring game} on a graph $G$ is a two-player game introduced by Brams \cite{Brams1981} and rediscovered ten years after by Bodlaender \cite{Bodlaender1991}. 
Given a set of $k$ colors, Alice and Bob take turns coloring properly an uncolored vertex of $G$, Alice having the first move.
Alice wins the game if all the vertices of $G$ are eventually colored, 
while Bob wins the game whenever, at some step of the game, all the colors appear in the neighborhood
of some uncolored vertex.
The \textit{game chromatic number} $\chi_g(G)$ of $G$ is then the smallest $k$ for which Alice has a winning strategy
when playing the graph coloring game on $G$ with $k$ colors.

The problem of determining the game chromatic number of planar graphs has attracted great interest in recent years.
Kierstead and Trotter proved in 1994 that every planar graph has game chromatic number at most 33~\cite{Kierstead1994}.
This bound was decreased to 30 by Dinski and Zhu~\cite{Dinski1999}, then to 19 by Zhu~\cite{Zhu1999}, to 18 by Kierstead~\cite{Kierstead2000}
and to 17, again by Zhu~\cite{Zhu2008}, in 2008. Some other classes of graphs have also been considered (see~\cite{survey} for
a comprehensive survey).

An \emph{incidence} of a graph $G$ is a pair $(v,e)$ where $v$ is a vertex of $G$ and $e$ an edge incident to $v$. We denote by $I(G)$ the set of incidences of $G$. Two incidences $(v,e)$ and $(w,f)$ are \emph{adjacent} if either (1) $v = w$, (2) $e = f$ or (3) $vw = e$ or $f$. 
An \emph{incidence coloring} of $G$ is a coloring of its incidences in such a way that adjacent incidences get distinct colors. The smallest number of colors required for an incidence coloring of $G$ is the \emph{incidence chromatic number} of $G$, denoted by $\chi_i(G)$. 
Let $G$ be a graph and $S(G)$ be the \emph{full subdivision} of $G$, obtained from $G$ by subdividing every edge of $G$ (that is, by replacing
each edge $uv$ by a path $ux_{uv}v$, where $x_{uv}$ is a new vertex of degree 2). It is then easy to observe that every incidence coloring of $G$ corresponds to a \emph{strong edge coloring} of $S(G)$, that is a proper edge coloring of $S(G)$ such that every two edges with the same color are at distance at least 3 from each other~\cite{BM93}. Observe also that any incidence coloring of $G$ is nothing but a distance-two coloring of the line-graph of $G$~\cite{IncidenceColoringPage}, that is a proper vertex coloring of the line-graph of $G$ such that any two vertices at distance two from each other get distinct colors.

Incidence colorings have been introduced by Brualdi and Massey~\cite{BM93} in 1993. 
Upper bounds on the incidence chromatic number have been proven for various classes of graphs such
as $k$-degenerate graphs and planar graphs~\cite{HS05,HSZ04}, graphs with maximum degree three~\cite{Maydansky}, 
and exact values are known for instance for forests~\cite{BM93}, $K_4$-minor-free graphs~\cite{HSZ04}, or
Halin graphs with maximum degree at least 5~\cite{WCP02} (see~\cite{IncidenceColoringPage} for an on-line survey).

In~\cite{Andres09}, Andres introduced the \emph{incidence coloring game}, as the incidence version of the graph coloring
game, each player, on his turn, coloring an uncolored incidence of $G$ in a proper way. The \emph{incidence game chromatic number} $i_g(G)$ of a graph $G$ is then defined as the smallest $k$ for which Alice has
a winning strategy when playing the incidence coloring game on $G$ with $k$ colors. Upper bounds on the incidence game chromatic number
have been proven for $k$-degenerate graphs~\cite{Andres09} and exact values are known for cycles, stars~\cite{Andres09}, paths and wheels~\cite{Kim10}.

Andres observed that the inequalities $\lceil\frac{3}{2}\Delta(G)\rceil \le i_g(F) \le 3\Delta(G) - 1$ hold for every graph $G$~\cite{Andres09}.
For $k$-degenerate graphs, he proved the following:

\begin{theorem}[Andres, \cite{Andres09}] \label{andres}
	Let $G$ be a $k$-degenerated graph. Then we have:
\begin{enumerate}[{\rm (i)}]
\item $i_g(G) \le 2\Delta(G) + 4k - 2$,
\item $i_g(G) \le 2\Delta(G) + 3k - 1$ if $\Delta(G) \ge 5k - 1$,
\item $i_g(G) \le \Delta(G) + 8k - 2$ if $\Delta(G) \le 5k - 1$.
\end{enumerate}
\end{theorem}

Since forests, outerplanar graphs and planar graphs are respectively 1-, 2- and 5-degenerate, we get that
$i_g(G)\le 2\Delta(G)+2$, $i_g(G)\le 2\Delta(G) + 6$ and $i_g(G)\le 2\Delta(G) + 18$ whenever $G$ is a forest, an outerplanar graph
or a planar graph, respectively.

Recall that the arboricity $a(G)$ of a graph $G$ is the minimum number of forests into which its set of edges
can be partitioned. In this paper, we will prove the following:

\begin{theorem} \label{th:general}
For every graph $G$, $i_g(G) \le \lfloor\frac{3\Delta(G) - a(G)}{2}\rfloor + 8a(G) - 1$.
\end{theorem}

Recall that $i_g(G)\ge 3\Delta(G)/2$ for every graph $G$ so that the difference between the upper and the lower bound
on $i_c(G)$ only depends on the arboricity of $G$.

It is not difficult to observe that $a(G) \le k$ whenever $G$ is a $k$-degenerate graph.
Hence we get the following corollary, which improves Andres' Theorem and answers in the negative a question posed in~\cite{Andres09}:

\begin{corollary}
If $G$ is a $k$-degenerate graph, then $i_g(G)\le \lfloor\frac{3\Delta(G) - k}{2}\rfloor + 8k - 1$
\end{corollary}

Since outerplanar graphs and planar graphs have arboricity at most 2 and 3, respectively,
we get as a corollary of Theorem~\ref{th:general} the following:

\begin{corollary} \mbox{}
\begin{enumerate}[{\rm (i)}]
\item $i_g(G)\le \lceil\frac{3\Delta(G)}{2}\rceil + 6$ for every forest $G$,
\item $i_g(G)\le \lfloor\frac{3\Delta(G)}{2}\rfloor + 14$ for every outerplanar graph $G$,
\item $i_g(G)\le \lceil\frac{3\Delta(G)}{2}\rceil + 21$ for every planar graph $G$.
\end{enumerate}
\end{corollary}

We detail Alice's strategy in Section~\ref{sec:strategy} and prove Theorem~\ref{th:general} in Section~\ref{sec:proof}. 



\section{Alice's Strategy}
\label{sec:strategy}

We will give a strategy for Alice which allows her to win the incidence coloring game on a graph $G$ with arboricity $a(G)$ whenever the number of available colors is at least  $\lfloor\frac{3}{2}\Delta(G)\rfloor + 8a(G) - 1$. This strategy will use the concept of {\em activation strategy}~\cite{survey}, often used in the context of the ordinary graph coloring game.

Let $G$ be a graph with arboricity $a(G)=a$. We partition the edges of $G$ into $a$ forests $F_1$, ..., $F_a$, each forest containing a certain number of trees. For each tree $T$, we choose an arbitrary vertex of $T$, say $r_T$, to be the \emph{root} of $T$.



\begin{notation}{\rm
Each edge with endvertices $u$ and $v$ in a tree $T$ will be denoted by $uv$ if $dist_T(u, r_T) < dist_T(v, r_T)$, and by $vu$ if $dist_T(v, r_T) < dist_T(u, r_T)$, where $dist_T$ stands for the distance within the tree $T$ (in other words, we define an orientation of the graph $G$ in such a way that all the edges of a tree $T$ are oriented from the root towards the leaves).
}\end{notation}

We now give some notation and definitions we will use in the sequel (these definitions are illustrated in Figure~\ref{fig}).
\begin{itemize}
\item 
For every edge $uv$ belonging to some tree $T$, we say the incidence $(u, uv)$ is a \emph{top incidence} whereas the incidence $(v, uv)$ is a \emph{down incidence}.  
We then let $t(uv)=t(v,vu)=(u,uv)$ and $d(uv)=d(u,uv)=(v,uv)$.\\
Note that each vertex in a forest $F_i$ is incident to at most one down incidence belonging to $F_i$, so that each vertex in $G$ is incident to at most $a$ down incidences. 

\item 

For every incidence $i$ belonging to some edge $uv\in E(G)$, let $tF(i)=\{t(wu),\ wu\in E(G)\}$ be the set of \emph{top-fathers} of $i$,
$dF(i)=\{d(wu),\ wu\in E(G)\}$ be the set of \emph{down-fathers} of $i$ and $F(i)=tF(i)\cup dF(i)$ be the set of \emph{fathers} of $i$.\\
Note that each incidence has at most $a$ top-fathers and at most $a$ down-fathers.

\item
For every incidence $i$ belonging to some edge $uv\in E(G)$, let $tS(i)=\{t(vw),\ vw\in E(G)\}$ be the set of \emph{top-sons} of $i$,
$dS(i)=\{d(vw),\ vw\in E(G)\}$ be the set of \emph{down-sons} of $i$ and $S(i)=tS(i)\cup dS(i)$ be the set of \emph{sons} of $i$.\\
Note that each incidence has at most $\Delta(G)-1$ top-sons and at most $\Delta(G)-1$ down-sons.

\item 
For every incidence $i$ belonging to some edge $uv\in E(G)$, let $tB(i)=\{t(uw),\ uw\in E(G)\} - \{i\}$ be the set of \emph{top-brothers} of $i$,
$dB(i)=\{d(uw),\ uw\in E(G)\} - \{i\}$ be the set of \emph{down-brothers} of $i$ and $B(i)=tB(i)\cup dB(i)$ be the set of \emph{brothers} of $i$.\\
Note that each top incidence $i$ has at most $\Delta(G)-|tF(i)|-1$ top-brothers and $\Delta(G)-|tF(i)|$ down-brothers while each down incidence $j$ has at most $\Delta(G)-|tF(j)|$ top-brothers and $\Delta(G)-|tF(j)|-1$ down-brothers.\\
Note also that any two brother incidences have exactly the same set of fathers.

\item 
Finally, for every incidence $i$ belonging to some edge $uv\in E(G)$, let $tU(i)=\{t(wv),\ wv\in E(G)\}$ be the set of \emph{top-uncles} of $i$,
$dU(i)=\{d(wv),\ wv\in E(G)\}$ be the set of \emph{down-uncles} of $i$ and $U(i)=tU(i)\cup dU(i)$ be the set of \emph{uncles} of $i$ (the term "uncle" is not metaphorically correct since the uncle of an incidence $i$ is another father of the sons of $i$ rather than a brother of a father of $i$).\\
Note that each incidence has at most $a-1$ top-uncles and at most $a-1$ down-uncles.
Moreover,  we have $|dU(i)| + |tS(i)| \le \Delta(G) - 1$ for every incidence $i\in I(G)$.
\end{itemize}

\begin{figure}
\begin{center}
  \includegraphics[width=8cm]{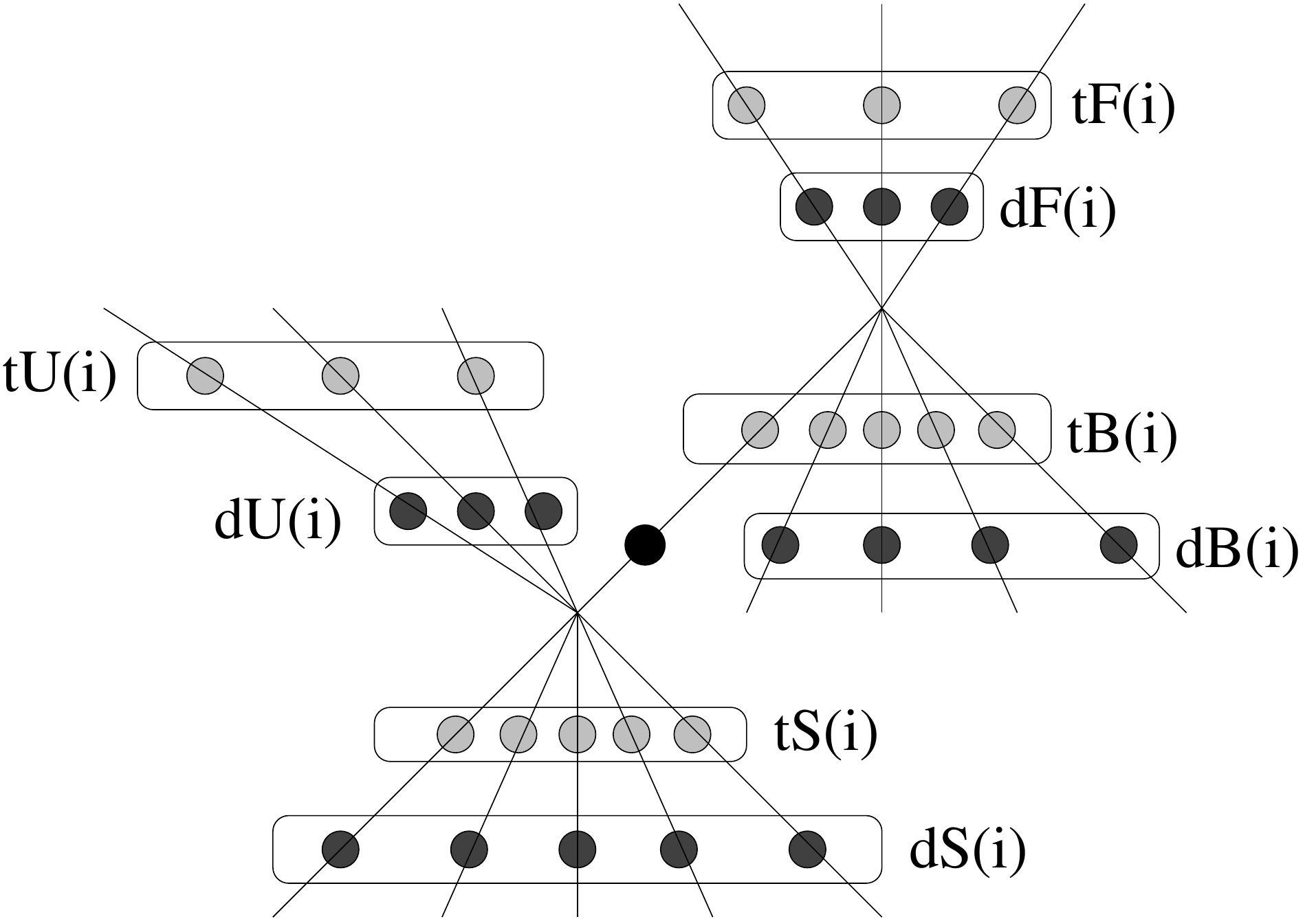}
\end{center}
	\caption{Incidences surrounding the incidence $i$.}
	\label{fig}
\end{figure}

Figure~\ref{fig} illustrates the above defined sets of incidences. Each edge is drawn in such a way that its top incidence is located above its down incidence. Incidence $i$ is drawn as a white box, top incidences are drawn as grey boxes and  down incidences (except $i$) are drawn as black boxes.

We now turn to the description of Alice's strategy.
For each set $I$ of incidences, we will denote by $I_c$ the set of colored incidences of $I$.
We will use an \emph{activation strategy}. During the game, each uncolored incidence may be either {\em active} (if Alice activated it) or {\em inactive}. When the game starts, every incidence is inactive.
When an active incidence is colored, it is no longer considered as active. In our strategy, it is sometimes possible that we say Alice activates an incidence that is already colored, 
but still considered as inactive.
For each set $I$ of incidences, we will denote by $I_a$ the set of active incidences of $I$ ($I_a$ and $I_c$ are therefore disjoint for every set of incidences $I$).

We denote by $\Phi$ the set of colors used for the game, by $\phi(i)$ the color of an incidence $i$ and, for each set $I$ of incidences, we let $\phi(I) = \bigcup_{i \in I} \phi(i)$.
As shown by Figure~\ref{fig}, the set of {\em forbidden} colors for an uncolored incidence $i$ is given by:
\begin{itemize}
\item $\phi(F(i)\cup B(i)\cup tS(i)\cup dU(i))$ if $i$ is a top incidence,
\item $\phi(dF(i)\cup tB(i)\cup S(i)\cup U(i))$ if $i$ is a down incidence.
\end{itemize}

Our objective is therefore to bound the cardinality of these sets.
We now define the subset $I_n$ of {\em neutral incidences} of $I(G)$, which contains all the incidences $j$ such that:
\begin{enumerate}[{\rm (i)}]
\item $j$ is not colored,
\item all the incidences of $F(j)$ are colored.
\end{enumerate} 

We also describe what we call a {\em neutral move} for Alice, that is a move Alice makes only if there is no neutral incidence and no activated incidence in the game. 
Let $i_0$ be any uncolored incidence of $I(G)$. 
Since there is no neutral incidence, either there is an uncolored incidence $i_1$ in $dF(i_0)$, or all the incidences of $dF(i_0)$ are colored and there is an uncolored incidence $i_1$ in $tF(i_0)$. We define in the same way incidences $i_2$ from $i_1$, $i_3$ from $i_2$, and so on, until we reach an incidence that has been already encountered. We then have $i_k=i_\ell$ for some integers $k$ and $\ell$, with $k \le \ell$. 
The neutral move of Alice then consists in activating all the incidences within the loop and coloring any one of them.

Alice's strategy uses four rules. The first three rules, (R1), (R2) and (R3) below, determine  which incidence Alice colors at each move.
The fourth rule explains which color will be used by Alice when she colors an incidence. 
\begin{enumerate}[(R1)]
	\item On her first move, 
	\begin{itemize}
		\item If there is a neutral incidence (i.e., in this case, an incidence without fathers), then Alice colors it.
		\item Otherwise, Alice makes a neutral move. 
	\end{itemize}
	\item If Bob, in his turn, colors a down incidence $i$ with no uncolored incidence in $dF(i)$, then 
	\begin{enumerate}[(R2.2.1)]
			\item If there are uncolored incidences in $dB(i)$, then Alice colors one of them,
			\item Otherwise,
				\begin{itemize}
					\item If there is a neutral incidence or an activated incidence in $I(G)$, then Alice colors it,
					\item If not, Otherwise, Alice makes a neutral move.
				\end{itemize}
		\end{enumerate}
	\item If Bob colors another incidence, then Alice \emph{climbs} it. Climbing an incidence $i$
                   is a recursive procedure, described as follows:
	\begin{enumerate}[(R3.1)]
		\item If $i$ is active, then Alice colors $i$.
		\item Otherwise, Alice activates $i$ and:
		\begin{itemize}
		\item If there are uncolored incidences in $dF(i)$, then Alice climbs one of them.
		\item If all the incidences of $dF(i)$ are colored, and if there are uncolored incidences in $tF(i)$, then Alice climbs one of them.
		\item If all the incidences of $F(i)$ are colored, then:
	                       \begin{itemize}
					\item if there is a neutral incidence or an activated incidence in $I(G)$, then Alice colors it,
					\item otherwise, Alice makes a neutral move.
				\end{itemize}
		\end{itemize}
	\end{enumerate}
        \item When Alice has to color an incidence $i$, she proceeds as follows: if $i$ is a down incidence with $|\phi(dB(i))| \ge 4a - 1$, she uses any available color in $\phi(dB(i))$; in all other cases, she chooses any available color.
\end{enumerate}

Observe that, in a neutral move, all the incidences $i_k,i_{k+1},\dots,i_\ell$ form a {\em loop} where each incidence can be reached by climbing the previous one. We consider that, when Alice does a neutral move, all the incidences are climbed at least one. 

Then we have:

\begin{observation} \label{obs:climbing}
When an inactive incidence is climbed, it is activated.
When an active incidence is climbed, it is colored.
Therefore, every incidence is climbed at most twice. 
\end{observation}

\begin{observation} \label{obs:coloring}
Alice only colors neutral incidences or active incidences (typically, incidences colored by Rule (R2.2.1) are neutral incidences), except when she makes a neutral move. 
\end{observation}

\section{Proof ot Theorem \ref{th:general}}
\label{sec:proof}

We now prove a series of lemmas from which the proof of Theorem~\ref{th:general} will follow.

\begin{lemma} \label{boundS}
When Alice or Bob colors a down incidence $i$, we have $$|S_c(i)| + |U_c(i)| \le 4a - 2.$$
When Alice or Bob colors a top incidence $i$, we have $$|tS_c(i)| + |dU_c(i)| \le 5a - 1.$$
\end{lemma}

\begin{proof}
Let first $i$ be a down incidence that has just been colored by Bob or Alice.
If $|S_c(i)|=0$, then $|S_c(i)| + |U_c(i)| = |U_c(i)| \le |U(i)| \le 2a-2$.
Otherwise, let $j$ be an incidence from $S(i)$ which was colored before $i$.
\begin{itemize}
	\item If $j$ was colored by Bob, then Alice has climbed $i$ or some other incidence from $dU(i)$ in her next move by Rule (R2.1). 
	\item If $j$ was colored by Alice, then
\begin{itemize}
        \item either $j$ was an active incidence and, when $j$ has been activated, Alice has climbed either $d(i)$, or $i$, or some other incidence from $U(i)$,
	\item or Alice has made a neutral move and, in the same move, has activated either $d(i)$, or $i$, or some other incidence from $U(i)$.
\end{itemize}
\end{itemize}
By Observation~\ref{obs:climbing} every incidence is climbed at most twice, and thus $|S_c(i)| \le 2 \times (|dU(i)| + 1)$.
Since $|dU(i)| \le a-1$, we have $|S_c(i)| \le 2a$. Moreover, since $|U_c(i)|  \le |U(i)| \le 2a - 2$, we get $|S_c(i)| +|U_c(i)| \le 4a - 2$ as required. 


Let now $i$ be a top incidence that has just been colored by Bob or Alice.
If $|tS_c(i)|=0$, then $|tS_c(i)| + |dU_c(i)| = |dU_c(i)| \le |dU(i)| \le a-1$.
Otherwise, let $j$ be an incidence from $tS(i)$ which was colored before $i$. 
\begin{itemize}
	\item If $j$ was colored by Bob then, in her next move, Alice either has climbed $d(i)$ or some other incidence from $dU(i)$ by Rule (R2.1),
 or $i$ or some other incidence from $tU(i)$ by Rule (R2.3).
	\item If $j$ was colored by Alice, then
\begin{itemize}
        \item either $j$ was an active incidence and, when $j$ has been activated, Alice has climbed either $d(i)$, or $i$, or some other incidence from $U(i)$,
	\item or Alice has made a neutral move and, in the same move, has activated either $d(i)$, or $i$, or some other incidence from $U(i)$.
\end{itemize}
\end{itemize}
By Observation~\ref{obs:climbing} every incidence is climbed at most twice, and thus $|tS_c(i)| \le 2 \times (|U(i)| + 2)$.
Since $|U(i)|\le 2a-2$, we have $|tS_c(i)| \le 4a$.
Moreover, since $|dU_c(i)|  \le |dU(i)| \le a - 1$, we get $|tS_c(i)| +|dU_c(i)| \le 5a - 1$ as required.
\end{proof}

\begin{lemma} \label{dBavailable}
Whenever Alice or Bob colors a down incidence $i$, there is always an available color for $i$ if $|\Phi| \ge \Delta(G) + 5a - 2$. Moreover, if $|\phi(dB(i))| \ge 4a - 1$, then there is always an available color in $\phi(dB(i))$ for coloring $i$.
\end{lemma}

\begin{proof}
When Alice or Bob colors a down incidence $i$, the forbidden colors for $i$ are the colors of $tB(i)$, $dF(i)$, $S(i)$ and $U(i)$. 

Observe that $|dF(i)| + |tB(i)| \le \Delta(G) - 1$ for each down incidence $i$, so $|\phi(dF(i))| + |\phi(tB(i))| \le \Delta(G) - 1$.

Now, since $|\phi(S(i) )| + |\phi(U(i))| \le |S_c(i)| + |U_c(i)| \le 4a - 2$ by Lemma~\ref{boundS}, we get that there are at most $\Delta(G) + 5a - 3$ forbidden colors, and therefore an available color for $i$ whenever $|\Phi| \ge \Delta(G) + 5a - 2$. 

Moreover, since the colors of $\phi(dF(i))$ and $\phi(tB(i))$ are all distinct from those of $\phi(dB(i))$, there are at most 
$|S_c(i)| + |U_c(i)| \le 4a - 2$ colors of $\phi(dB(i))$ that are forbidden for $i$, and therefore an available color for $i$ whenever $|\phi(dB(i))| \ge 4a - 1$.
\end{proof}

\begin{lemma} \label{bounddB}
For every incidence $i$, $|\phi(dB(i))| \le \lfloor\frac{|dB(i)|}{2}\rfloor + 2a$.
\end{lemma}

\begin{proof}
For every incidence $i$, as soon as $|\phi(dB(i))| = 4a - 1$, there are at least $4a-1$ colored incidences in $dB(i)$.
If $dF(i)$ is not empty, then every incidence in $dF(i)$ has thus at least $4a-1$ colored sons so that, by  Lemma~\ref{boundS}, every such incidence is already colored.
During the rest of the game, each time Bob will color an incidence of $dB(i)$, if there are still some uncolored incidences in $dB(i)$, then Alice will answer by coloring one of them by Rule (R2.2.1). Hence, Bob will color at most $\lceil\frac{|dB(i) - (4a - 1)|}{2}\rceil$ of these incidences. Since, by Rule~(R3), Alice uses colors already in $\phi(dB(i))$ for the incidences she colors, we get $|\phi(dB(i))| \le 4a - 1 + \lceil\frac{|dB(i) -(4a - 1)|}{2}\rceil \le \lfloor\frac{|dB(i)|}{2}\rfloor+ 2a$ as required.
\end{proof}

\begin{lemma}\label{top}
When Alice or Bob colors a top incidence $i$, there is always an available color for $i$ whenever $|\Phi| \ge \lfloor\frac{3\Delta(G) - a}{2}\rfloor + 8a - 1$.
\end{lemma}

\begin{proof}
Let $i$ be any uncolored top incidence. The forbidden colors for $i$ are the colors of $tF(i)$, $dF(i)$, $tB(i)$, $dB(i)$, $dU(i)$ and $tS(i)$. We have:
\begin{itemize}
	\item $|\phi(tF(i))| + |\phi(tB(i))| \le |tF(i)| + |tB(i)| \le \Delta(G) - 1$,
	\item $|\phi(dF(i))| \le |dF(i)| \le a$ and, by Lemma~\ref{bounddB}, $|\phi(dB(i))| \le \lfloor\frac{|dB(i)|}{2}\rfloor + 2a$; 
since $|dF(i)| + |dB(i)| \le \Delta(G)$, we get
\begin{center}\begin{tabular}{rcl}
$|\phi(dF(i))| + |\phi(dB(i))|$ & $\le$ & $|dF(i)| + \lfloor\frac{\Delta(G) - |dF(i)|}{2}\rfloor + 2a$ \\
        & $=$ & $\lceil\frac{3|dF(i)|}{2}\rceil + \lfloor\frac{\Delta(G)}{2}\rfloor + 2a$ \\
        & $\le$ & $\lceil\frac{3a}{2}\rceil + \lfloor\frac{\Delta(G)}{2}\rfloor + 2a$ \\
        & $=$ & $\lfloor\frac{\Delta(G) - a}{2}\rfloor + 3a$,
\end{tabular}\end{center}
	\item $|\phi(tS(i))| + |\phi(dU(i))| \le 5a - 1$ by Lemma~\ref{boundS}.
\end{itemize}
So there are at most $\lfloor\frac{3\Delta(G) - a}{2}\rfloor + 8a - 2$ forbidden colors for $i$ and the result follows. 
\end{proof}

Lemma~\ref{top} shows that, when Alice applies the strategy above described, every top incidence can be colored, provided
$|\Phi| \ge \lfloor\frac{3\Delta(G) - a}{2}\rfloor + 8a - 1$. Lemma~\ref{dBavailable} shows that this is also the case for
down incidences, which proves Theorem~\ref{th:general}.

\bibliographystyle{plain}
\bibliography{ICGA} 

\end{document}